\documentclass[onecolumn,showpacs]{revtex4-2}
\usepackage{xcolor}
\usepackage{xcolor}
\usepackage{graphics}
\usepackage{graphicx}
\usepackage{tikz-cd}
\usepackage{dcolumn}
\usepackage{amsmath}
\usepackage{amssymb}
\usepackage{mathtools}
\usepackage{hyperref}

\newtheorem{theorem}{Theorem}[section]
\newtheorem{corollary}[theorem]{Corollary}
\newtheorem{proposition}[theorem]{Proposition}

\newenvironment{proof}[1][Proof]{\begin{trivlist}
\item[\hskip \labelsep {\bfseries #1}]}{\end{trivlist}}
\allowdisplaybreaks

\newenvironment{remark}[1][Remark]{\begin{trivlist}
\item[\hskip \labelsep {\bfseries #1}]}{\end{trivlist}}

\newcommand{\qed}{\nobreak \ifvmode \relax \else
	\ifdim\lastskip<1.5em \hskip- \lastskip
	\hskip 0.5em plus0em minus0.5em \fi \nobreak
	\vrule height0.75em width0.5em depth0.25em\fi}

\begin{document}

\title{Horizon area bound and MOTS stability in locally rotationally symmetric solutions}

\author{Abbas M. \surname{Sherif}}
\email{abbasmsherif25@ibs.re.kr}
\affiliation{Center for Geometry and Physics, Institute for Basic Science (IBS), Pohang 37673, Korea}

\author{Peter K. S. \surname{Dunsby}}
\email{peter.dunsby@uct.ac.za}
\affiliation{Cosmology and Gravity Group, Department of Mathematics and Applied Mathematics, University of Cape Town, Rondebosch 7701, South Africa\\South African Astronomical Observatory, Observatory 7925, Cape Town, South Africa}

\begin{abstract}
In this paper, we study the stability of marginally outer trapped surfaces (MOTS), foliating horizons of the form \(r=X(\tau)\), embedded in locally rotationally symmetric class II perfect fluid spacetimes. An upper bound on the area of stable MOTS is obtained. It is shown that any stable MOTS of the types considered in these spacetimes must be strictly stably outermost, that is, there are no MOTS ``outside" of and homologous to \(\mathcal{S}\). Aspects of the topology of the MOTS, as well as the case when an extension is made to imperfect fluids, are discussed. Some non-existence results are also obtained. Finally, the ``growth" of certain matter and curvature quantities on certain unstable MOTS are provided under specified conditions. 
\end{abstract} 

\maketitle

\tableofcontents

\section{Introduction}
\label{intro}


Stability of marginally outer trapped surfaces (MOTS) (see the references \cite{haw1,haw2,rpac,and1,and2,gal1,gal2}, or more recently, \cite{ib1,ib2,ib3}, where other interpretations of stability have allowed direct comparisons to the geodesic Jacobi operator), a Lorentzian analogue of the well studied notion of stability of minimal surfaces in Riemannian geometry, has generated significant interest over the years due to its utility is proving many useful results in general relativity. A relatively recent interpretation of MOTS stability, \cite{gal2}, has proven very successful in proving higher dimensional results for the topology of black holes, which relies on showing that under certain assumptions on either the initial data set containing the MOTS or the spacetime, the MOTS is of positive Yamabe type, i.e., admits a metric of constant positive scalar curvature. 

Area bounds for black holes have also been obtained under different scenarios. Dain \textit{et al.}, \cite{sd1,sd2,sd3}, obtained lower bounds on the area in terms of the charges and angular momenta. Andersson \textit{et al.} \cite{and3}, again, making use of the stable properties of the MOTS \(\mathcal{S}\), obtained that the area of the black hole either has an upper bound which is a constant function of the curvature, injectivity radius and the volume of the initial data set containing \(\mathcal{S}\) or, another MOTS \(\mathcal{S}'\) lie to the outside of \(\mathcal{S}\). The proof of this area inequality, in particular the gluing approach used, shows that two stably outermost MOTS that are sufficienctly close will have a MOTS surrounding them. This  has been interpreted as demonstrating a well known fact of black hole mergers: two black holes that are sufficiently close will merge (see the reference \cite{andt}).

The main objective of this paper is to study the stability of MOTS in locally rotationally symmetric (LRS) perfect fluids (though in certain cases the results do extend to the case of anisotropic fluids with non-vanishing heat flux) and examine its relationship to upper bounds on the horizon area, as well as the implications for the topology of the MOTS. In doing this, we consider additional restrictions on some geometric quantities. 

Our interest is primarily in dynamical horisons (DH) and/or timelike membranes (TLM), and as was demonstrated in \cite{shef1}, only the class II LRS solutions (those with neither rotation nor spatial twist) admit DH or TLM of a certain functional form. These will therefore be the focus of this work. 

LRS II solutions are vorticity-free solutions, each point of which admits a continuous isotropy group \cite{gebb,ge1}. The LRS class of solutions generalizes the well studied spherically symmetric solutions of the Einstein's field equations. The line element of these solutions carries the specific form

\begin{eqnarray*}
ds^2=-C^2_1d\tau^2+C^2_2dr^2+C^2_3\left(dy^2+D^2dz^2\right),
\end{eqnarray*}
where \(C_{1,2,3}\) are functions of \(\tau\) and \(r\), and \(D\) is a function of \(y\) and a parameter \(k\) which parametrizes the 2-surfaces of constant \(r\) in the \(\tau=\) constant slices: the 2-surface geometry is spherical for \(k=+1\), hyperbolic for \(k=-1\), and flat for \(k=0\). (The \(k=+1\) is the spherically symmetric subclass.)

The approach to be employed in this work is a covariant (and hence gauge invariant) approach, known as the 1+1+2 covariant formalism (see \cite{cc1} and references therein). This approach has recently been used to study horizons in LRS spacetimes \cite{shef2}. Like the powerful 1+3 formalism which allows for the threading of the spacetime along the fluid flow lines with unit tangent vector \(u^{\mu}\), this approach is a specialization, where the 3-space orthogonal to the fluid flow is further threaded along the vorticity, with this direction specified by the unit vector field \(e^{\mu}\). This allows for the projection of tensor and vector quantities, as well as derivatives (a dot \(\dot{\ }\) denotes derivative along \(u^{\mu}\), a hat \(\hat{\ }\), that along \(e^{\mu}\), and \(\delta_{\mu}=N^{\nu}_{\mu}\nabla_{\nu}\) the derivative on the 2-space which results from decomposing the 3-space, along these directions, with \(N^{\mu\nu}\) projecting vectors and tensors orthogonal to \(u^{\mu}\) and \(e^{\mu}\), to the 2-space.) The derivative operator \(\nabla_{\mu}\) denotes the 4-dimensional spacetime covariant derivative. The form of the field equations \(G_{\mu\nu}=T_{\mu\nu}\) is used instead of \(G_{\mu\nu}=8\pi T_{\mu\nu}\).

In regards to horizons, the splitting, when considering the background (i.e. exact) solutions, singles out surfaces in the foliation along \(u^{\mu}\). The particular surfaces we consider are picked out by the introduction of the frame vector \(e^{\mu}\). In our particular case in this work, the split here singles out a ``preferred" horizon. It must however be emphasized that, the 1+1+2 split works for perturbed LRS (non-exact) solutions, in which case we cannot claim that the split picks up a preferred horizon, since the approach is not frame-invariant.  Different frame choices in the perturbed solution will single out different horizons.

For LRS II solutions, one can characterize the spacetimes entirely by the set of covariant scalars \cite{cc1,cc2,cc3}

\begin{eqnarray*}
\mathcal{D}:\equiv\lbrace{\rho,p,Q,\Pi,\mathcal{E},\mathcal{A},\Theta,\Sigma,\phi\rbrace},
\end{eqnarray*}
where \(\rho\equiv T_{\mu\nu}u^{\mu}u^{\nu}\) is the energy density, \(p\equiv\left(1/3\right)h^{\mu\nu}T_{\mu\nu}\) is the pressure (isotropic), \(Q=-T_{\mu\nu}e^{\mu}u^{\nu}\) is the heat flux, \(\Pi=T_{\mu\nu}e^{\mu}e^{\nu}-p\) is the anisotropic stress, \(\mathcal{E}=E_{\mu\nu}e^{\mu}e^{\nu}\) encodes the electric part of the Weyl tensor \(E_{\mu\nu}\), \(\mathcal{A}=\dot{u}_{\mu}e^{\mu}\) is the acceleration, \(\Theta\equiv D_{\mu}u^{\mu}\) is the expansion, \(\Sigma=e^{\mu}e^{\nu}D_{\langle \nu}u_{{\mu}\rangle}\) is the shear and \(\phi=\delta_{\mu}e^{\mu}\) denotes the expansion of the \(2\)-space (referred to as the sheet expansion), with

\begin{eqnarray*}
T_{\mu\nu}&=\rho u_{\mu}u_{\nu} + 2Qe_{(\mu}u_{\nu)} +\left(p+\Pi\right)e_{\mu}e_{\nu} + \left(p-\frac{1}{2}\Pi\right)N_{\mu\nu},
\end{eqnarray*}
being the stress energy tensor, and the derivative operator \(D_{\mu}\) denoting the covariant derivative on the hypersurface to which \(u^{\mu}\) is hypersurface orthogonal. 

The LRS II class of spacetimes can be extended to the more general LRS class by including rotation \(\Omega=\omega^ae_a\) (\(\omega_a\) is the vorticity vector) and spatial twist (the twist of \(e^{\mu}\)) \(\xi=(1/2)\varepsilon^{\mu\nu}\delta_{\mu}e_{\nu}\), as well as \(\mathcal{H}=H_{\mu\nu}e^{\mu}e^{\nu}\) (encoding the magnetic part of the Weyl tensor), in the set \(\mathcal{D}\). (The use of the word vorticity should not be confused with angular momentum)

While in general, for a 1+1+2 decomposed spacetimes, the 2-sheets are not necessarily true surfaces, for LRS II solutions this is true as is seen from the fact that the relation \(\delta_{\mu}\delta_{\nu}\psi=\delta_{\nu}\delta_{\mu}\psi\) holds when acting on an arbitrary scalar \(\psi\) (and trivially so for LRS scalars in the covariant set \(\mathcal{D}\)). In adddition, \(u^{\mu}\) and \(e^{\mu}\) are surface forming since their commutator has no sheet component \cite{cc1,cc2}:

\begin{eqnarray*}
\hat{\dot{\psi}}-\dot{\hat{\psi}}=-\mathcal{A}\dot{\psi}+\left(\frac{1}{3}\Theta+\Sigma\right)\hat{\psi}.
\end{eqnarray*}
For more details, the interested reader is referred to \cite{cc1}.

The field equations for these spacetimes can be obtained using the Ricci identities for the preferred unit directions and the contracted Ricci identities (further details can be found in \cite{cc1}). Throughout this work, whenever needed, we will provide the relevant equations. 

This paper has the following structure: In section \ref{5}, we briefly introduce the notion of MOTS in LRS spacetimes, formulated in the covariant approach employed throughout this paper. A marginally outer trapped tube (MOTT) is also introduced and the nature of their evolution discussed. In Section \ref{mr}, a short overview of the notion of stability of MOTS is given. It is then applied to perfect fluid LRS II spacetimes, and an upper bounds on the area are provided as being necessary and/or sufficient for stability of the MOTS. The topology of the MOTS are then briefly considered as it relates to bound on the horizon area, with accompanying discussions. Section \ref{bound} considers bounds on growth of some of the spacetime variables in the neighborhood of certain unstable MOTS. The results are discussed in Section \ref{8} and ongoing work where the results of this paper would prove useful are also mentioned.


\section{MOTS and MOTT in LRS solutions}\label{5}


Let \(\mathcal{S}\) (assumed to be spacelike and closed) be an embedded \(2\)-surface in a LRS spacetime, where \(u^{\mu}\) and \(e^{\mu}\) are both normal to \(\mathcal{S}\), with \(u_{\mu}u^{\mu}=-1,e_{\mu}e^{\mu}=1\) and \(e_{\mu}u^{\mu}=0\). Then, the respective tangents to the outgoing and ingoing null geodesics are given by (see \cite{rit1})

\begin{eqnarray*}
\begin{split}
k^{\mu}&=\frac{1}{\sqrt{2}}\left(u^{\mu}+e^{\mu}\right),\\
l^{\mu}&=\frac{1}{\sqrt{2}}\left(u^{\mu}-e^{\mu}\right).
\end{split}
\end{eqnarray*}

(We note that, in general, there is a degree of freedom to scale the null vectors as \(k^{\mu}\rightarrow fk^{\mu}\) and \(l^{\mu}\rightarrow f^{-1}l^{\mu}\), for \(f>0\). However, such detail will not be relevant to the rest of this paper. This will be seen shortly as we discuss the causal character of horizons.  Furthermore, in the case of the stability operator that will be introduced later, it was shown in \cite{jl1} that the spectrum of the eigenvalues remain invariant under such scaling.) The divergences of the congruences generated by \(k^{\mu}\) and \(l^{\mu}\), also known as the null expansion scalars with respect to the null directions, are calculated respectively as \cite{shef2}:

\begin{eqnarray*}
\begin{split}
\chi^+&=\chi^+\left(\tau,r\right)=\frac{1}{\sqrt{2}}\left(\frac{2}{3}\Theta-\Sigma+\phi\right),\\
\chi^-&=\chi^-\left(\tau,r\right)=\frac{1}{\sqrt{2}}\left(\frac{2}{3}\Theta-\Sigma-\phi\right),
\end{split}
\end{eqnarray*}
which are just the traces of the null second fundamental forms associated with \(k^{\mu}\) and \(l^{\mu}\) respectively, where we have labelled as \(\tau\) and \(r\) the parameters of the integral curves along the vector fields \(u^{\mu}\) and \(e^{\mu}\). \(\mathcal{S}\) is then said to be \textit{marginally outer trapped} (resp. \textit{outer trapped}) if \(\chi^+=0\)  (resp. \(\chi^+<0\)) \cite{ib4,ash1,ash2}. (Henceforth, we will simply write \(\chi=\chi^+\).) In cases where both expansions vanish, \(\mathcal{S}\) is said to be minimal, a standard case in point being cross sections of the Schwarzschild horizon in isotropic coordinates (notice that for Schwarzschild, \(\Theta=\Sigma=0\), so that \(\chi^+=\phi=0\implies\chi^-=-\phi=0\)). A \textit{marginally outer trapped tube} (MOTT) is then a hypersurface foliated by MOTS. In the case that \(\chi^-<0\), the ``outer" is dropped and the notations MTS and MTT are instead used (see for example \cite{ib5}). From the expressions for \(\chi\) and \(\chi^-\), it is clear that if \(\chi\) vanishes, then, one requires that \(\phi>0\) in order for \(\chi^-<0\) on \(\mathcal{S}\). 

For a MOTT, one can make the choice of a vector field \cite{ib5,ib6}

\begin{eqnarray*}
\mathcal{V}^{\mu}=k^{\mu}-\mathcal{C}l^{\mu},
\end{eqnarray*}
assuming that \(\mathcal{V}^{\mu}\) is tangent to the MOTT and everywhere normal to \(\mathcal{S}\), for some smooth function \(\mathcal{C}\) on the MOTT. Then, there exists a unique vector field \(\tilde{\mathcal{V}}^{\mu}\) which is normal to the MOTT, and is given as

\begin{eqnarray*}
\tilde{\mathcal{V}}^{\mu}=k^{\mu}+\mathcal{C}l^{\mu}.
\end{eqnarray*}
It is easily seen that the sign of \(\mathcal{C}\) specifies the signature of the induced metric on the MOTT. 

While in general the sign of \(\mathcal{C}\) may vary at different points of the MOTT, we are interested in those cases where the sign of \(\mathcal{C}\) remains fixed on the MOTT. For a MOTT which is also an MTT (or alternatively a MOTT with positive sheet expansion), the MOTT is called a dynamical horizon (DH) if \(\mathcal{C}>0\), a timelike membrane (TLM) if \(\mathcal{C}<0\), and an isolated horizon (IH) if \(\mathcal{C}=0\) (in this case the sheet expansion may be zero). Indeed it is necessarily true that, on DH, IH and TLM the sheet expansion is non-negative, with the sheet expansion vanishing if and only if the MTS is minimal.

In the rest of this paper we will be dealing with situations where \(\mathcal{C}\) is constant on each MOTS. In that case, the scalar \(\mathcal{C}\) can be explicitly computed using the formula \cite{ib5} 

\begin{eqnarray}\label{x1}
\mathcal{C}=\frac{\mathcal{L}_k\chi}{ \mathcal{L}_l\chi},
\end{eqnarray}
which, for LRS II spacetimes is given as \cite{shef2}

\begin{eqnarray}\label{x2}
\mathcal{C}=\frac{-\left(\rho+p+\Pi\right)+2Q}{\frac{1}{3}\left(\rho-3p\right)+2\mathcal{E}},
\end{eqnarray}
where \(\mathcal{L}_k\) (\textit{resp.} \(\mathcal{L}_l\)) denotes the Lie derivative along the vector field \(k^{\mu}\) (\textit{resp.} \(l^{\mu}\)). If the null energy condition holds, then \(\mathcal{L}_k\chi\leq0\) so that the sign of the metric signature is specified by the denominator of \eqref{x2}, a curvature condition (see \cite{shef2} for more details).

The term \text{horizon} will be used to refer to MOTT (or MTT), and we will denote by \(\mathcal{H}\), whenever we refer to a DH. MOTS and MOTT will be consistently used, and whenever there is ambiguity, clarity will be provided.

Of interest throughout the rest of the paper will be those TLM and DH of the form \(r=X(\tau)\) in the class of perfect fluid solutions with no anisotropy and vanishing heat flux. It will be assumed throughout that, (at least) for any point \(\bar{p}\) on a MOTT, and a neighborhood \(\mathcal{N}_{\bar{p}}\) of \(\bar{p}\), \(\rho\) is non-negative on \(\mathcal{N}_{\bar{p}}\). (It is possible that this might hold true at all points of the ambient spacetime.) in the case the \(\rho=0\) applies, this will be made explicit. Otherwise, it will be assumed that \(\rho\) is positive.


\section{Area bounds from stability and topological implications}\label{mr}


In this section we consider the relationship between stability of MOTS and the area of the horizons they foliate. The focus of this section will be LRS II perfect fluids, unless otherwise stated.

\subsection{MOTS stability}

The notion of stability for MOTS - analogous to the well understood notion of stability of minimal surfaces - was introduced by Anderson \textit{et al.} \cite{and1,and2}. Simply put, it addresses the following question: does the deformation of a MOTS \(\mathcal{S}\) along the unit normal direction \(e^{\mu}\) leave \(\mathcal{S}\) marginally outer trapped or does \(\mathcal{S}\) become untrapped? In other words, denote by \(\mathcal{S}_t\) the deformation of \(\mathcal{S}\), and let \(\partial/\partial_t=\Psi e^{\mu}\) denote the tangent vector to the curve which generates the variation, for some function \(\Psi\). Is the associated variation of the outward null expansion as

\begin{eqnarray*}
\chi(\Psi):\equiv\left.\frac{\partial}{\partial t}\right|_{t=0}\chi:\equiv\bar{\delta}_{\Psi e}\chi\geq0,
\end{eqnarray*} 
and positive somewhere on \(\mathcal{S}\), or is \(\chi(\Psi)<0\)? A MOTS \(\mathcal{S}\) is said to be \textit{stable} (or \textit{stably outermost}) if \(\chi(\Psi)\geq0\) (with \(\chi(\Psi)\not\equiv0\)), \textit{strictly stable} (or \textit{strictly stably outermost}) if \(\chi(\Psi)>0\) everywhere on \(\mathcal{S}\), and \textit{unstable} otherwise. The above mentioned authors reduced this problem to an eigenvalue problem by introducing a second order elliptic operator \(L_{\mathcal{S}}:C^{\infty}(\mathcal{S})\rightarrow C^{\infty}(\mathcal{S})\) acting on \(\Psi\) as \cite{and1,and2}

\begin{eqnarray}\label{ar1}
\chi(\Psi):\equiv L_{\mathcal{S}}\Psi= - \Delta\Psi+2s^{\mu}\delta_{\mu}\Psi+\biggl(\frac{1}{2}R_{\mathcal{S}}-\left(\rho+J_{\mu}e^{\mu}\right)-\frac{1}{2}\chi_{\mu\nu}\chi^{\mu\nu}+\delta_{\mu}s^{\mu}-s_{\mu}s^{\mu}\biggr)\Psi,
\end{eqnarray}
where \(\Delta\) is the Laplacian on \(\mathcal{S}\), \(R_{\mathcal{S}}\) is the scalar curvature of \(\mathcal{S}\), \(J_{\mu}=G_{\mu\nu}u^{\nu}\) is the local momentum density along the horizon (\(G_{\mu\nu}\) is the Einstein tensor), \(\chi_{\mu\nu}\) is the null second fundamental form associated with \(k^{\mu}\), and the one-form \(s_{\mu}=-(1/2)l_{\nu}\delta_{\mu}k^{\nu}\) is the torsion of the null normal field \(k^{\mu}\) projected to \(\mathcal{S}\).

In general, the operator \(L_{\mathcal{S}}\) is not self-adjoint, a problem which arises due to the presence of the second (linear) term in \eqref{ar1}. However, in certain specific cases - for example, if \(\mathcal{S}\) lies in a time-symmetric horizon (where \(L_{\mathcal{S}}\) reduces to that of minimal surfaces) or, in some cases where the one-form \(s_{\mu}\) is a gradient - \(L_{\mathcal{S}}\) is self-adjoint (this is discussed in \cite{and2,gal1} etc., and references mentioned therein). This means that the eigenvalue spectrum is generally complex. Nonetheless, there is always a \textit{principal eigenvalue} \(\lambda\), which is real, and never bigger than the real part of any other eigenvalue of \(L_{\mathcal{S}}\) such that

\begin{eqnarray}\label{ope}
L_{\mathcal{S}}\Psi=\lambda\Psi,
\end{eqnarray}
where \(\Psi\) can be chosen to be strictly positive. Then, \(\mathcal{S}\) being stably outermost is equivalent to the condition that \(\lambda\geq0\), and \(\mathcal{S}\) being strictly stably outermost is equivalent to the condition that \(\lambda>0\). In either case one chooses the associated eigenfunction as everywhere positive.

From our formulation, it is by now clear the fundamental role that the sheet expansion \(\phi\) plays in locating marginally trapped surfaces and the horizon evolution. In particular, whether we can specialize from a MOTS to a MTS (and hence from a MOTT to an MTT) depends on whether the unit normal \(e^{\mu}\) diverges or converges on the MOTS. What are the implications that such property of \(\phi\) has for stability of the MOTS?

Indeed, strict positivity of \(\lambda\implies\chi(\Psi)>0\) for positive \(\Psi\). Now, let us take the case of a MTS, considering the class of spacetimes in this work, and denote by \(\phi(\Psi)\) the associated variation of the sheet expansion on \(\mathcal{S}\). Suppose \(\chi(\Psi)=0\) on \(\mathcal{S}\). Clearly, if \(\phi(\Psi)\leq0\), then \(\chi^-(\Psi)\geq0\). Under which condition(s) will this affect the sign of \(\chi^-\)? Conversely, suppose that the variation \(\phi(\Psi)\geq0\) (determining the sign of \(\phi(\Psi)\) can be written as an eigenvalue problem in a similar manner as with the case of \(\chi(\Psi)\)). Does the associated eigenvalue - which we denote by \(\bar{\lambda}\) - carry sufficient information about the stability of \(\mathcal{S}\)? This could certainly be of interest. We will briefly return to this in a short while, but a detailed analysis of the situation will be deferred to a future work.

\begin{remark}[Remark 1.]
\textit{If one were to consider a more general 2-surface \(\underline{\mathcal{S}}\) in an arbitrary spacetime admitting a 1+1+2 decomposition, the null vector field \(k^{\mu}\) (\textit{resp}. \(l^{\mu}\)) acquires a sheet component - denote this by \(m^{\mu}\) (\textit{resp}. - \(m^{\mu}\)) - orthogonal to both \(u^{\mu}\) and \(e^{\mu}\). The induced metric on \(\underline{\mathcal{S}}\) can be decomposed as}

\begin{eqnarray*}
\bar{N}_{\mu\nu}=N_{\mu\nu}+P_{\mu\nu},
\end{eqnarray*}
\textit{due to how the spacetime is decomposed, where \(P_{\mu\nu}\) is some symmetric 2-tensor. The expansion scalars \(\chi\) and \(\chi^-\) will therefore acquire the additional terms \(w+P^{\mu\nu}\nabla_{\mu}k_{\nu}\) and \(-w+P^{\mu\nu}\nabla_{\mu}l_{\nu}\) (\(w\) is some scalar), respectively. As per the discussion in the previous paragraph, one would then consider whether the quantity}

\begin{eqnarray*}
(\phi(\Psi)+w(\Psi))+Z(\Psi),
\end{eqnarray*}
\textit{carries sufficient information to determine stability, where we have denoted by \(Z(\Psi)\) the scalar \(P^{\mu\nu}\nabla_{\mu}(e_{\nu}+m_{\nu})\) associated to \(\underline{\mathcal{S}}_t\).}
\end{remark}

For the rest of this work  we shall assume that \(\phi\geq0\), with equality holding only if \(\mathcal{S}\) is minimal. That is, when ``MOTS" is mentioned, it is understood that the consideration allows for a negative sheet expansion.

Now, for a MOTS \(\mathcal{S}\) embedded in a LRS II solution, \(L_{\mathcal{S}}\) is self-adjoint. This follows from the fact that the projection of \(s_{\mu}\) to \(\mathcal{S}\) is zero:

\begin{eqnarray*}
\begin{split}
s_{\mu}&=-\frac{1}{2}N^{\mu\nu}\left(\mathcal{A}u_{\mu}+\left(\frac{1}{3}\Theta+\Sigma\right)e_{\mu}\right)
&=0. 
\end{split}
\end{eqnarray*}
(This is consistent with the fact that this is true in spherically symmetric LRS solutions.) In fact, the vanishing of \(s_{\mu}\) will be true for 2-surfaces in a general LRS solutions with rotation and/or spatial twist. In this case, the one-form \(s_{\mu}\) will take the exact same form as above since the rotation and twist terms scale the area form \(\varepsilon_{\mu\nu}\), which is zero when contracted with the unit vectors.


\subsection{Area bounds}


Let us now consider some results which relate stability of a MOTS to the size of the area of the horizon in which the MOTS lie. Unless otherwise stated, all considerations are with respect to MOTS (or MTS) in LRS II solutions where the MOTS respect the LRS II symmetries. It will also be assumed that the isotropic pressure \(p\) is non-negative.

\begin{proposition}\label{prop1}
Let \(\mathcal{S}\) be a MOTS embedded in a perfect fluid. Then, \(\mathcal{S}\) is stably outermost if and only if the horizon area \(\mathbb{A}\) satisfies

\begin{eqnarray}\label{proeq1}
\mathbb{A}\leq\frac{1}{2\rho}.
\end{eqnarray}
\end{proposition}
\begin{proof}
The 4-covariant derivative in LRS spacetimes, pulled back to \(\mathcal{S}\), vanishes when acting on scalars (\(\nabla_{\mu}\psi=-\dot{\psi}u_{\mu}+\hat{\psi}e_{\mu}\) \(\forall\) smooth functions \(\psi\) in any LRS spacetime) so that the surface covariant derivative and hence the Laplacian vanish when acting on scalars respecting the LRS II symmetries. We note that ``scalars" in LRS spacetimes, in context of the formulation used, is understood to mean elements of the covariant set \(\mathcal{D}\). With respect to the current work, this means that our considerations, for the remainder of this work, are those principal eigenvalues with constant eigenfunctions \(\Psi\). 

Now, assume such positive \(\Psi\) exists as in \eqref{ope}. Then, the principal eigenvalue \(\lambda\), which is real, simply reduces to

\begin{eqnarray}\label{proeq2}
\lambda=\frac{1}{2}R_{\mathcal{S}}-\left(\rho+J_{\mu}e^{\mu}\right)-\frac{1}{2}\chi_{\mu\nu}\chi^{\mu\nu},
\end{eqnarray}
with \(R_{\mathcal{S}}=2K\), where \(K\) is the Gaussian curvature of \(\mathcal{S}\) given by (we note that \(\mathcal{S}\) is a MOTS) (see \cite{rit1})

\begin{eqnarray*}
K=\frac{1}{3}\rho-\mathcal{E}.
\end{eqnarray*}
The local momentum density and the null second fundamental form associated to \(k^{\mu}\) are given by 

\begin{eqnarray*}
\begin{split}
J_{\mu}&=-\left(\rho u_{\mu}+Qe_{\mu}\right),\\
\chi_{\mu\nu}&=\frac{1}{2}\chi N_{\mu\nu}\overset{\mathrm{\mathcal{S}}}{=}0,
\end{split}
\end{eqnarray*}
so that for perfect fluids, \(J_{\mu}e^{\mu}=0\). Explicitly, \eqref{proeq2} takes the form

\begin{eqnarray*}
\lambda=-\frac{2}{3}\rho-\mathcal{E}. 
\end{eqnarray*}

(Note that the first term is non-positive by assumption.) Now, the Lie derivative of \(\chi\) along \(l^{\mu}\) can be written as \cite{ib5}

\begin{eqnarray}\label{proeq3}
\rho-p-\frac{1}{\mathbb{A}},
\end{eqnarray}
up to a factor of \(2\pi\) (which is \(1/4\) in the natural units), which upon comparing to the denominator of \eqref{x2} (also up to a factor of \(2\pi\)) gives

\begin{eqnarray}\label{proeq4}
\mathbb{A}=\frac{3}{2\left(\rho-3\mathcal{E}\right)}.
\end{eqnarray}
Therefore, if \(\lambda\geq0\), then, the inequality \eqref{proeq1} follows from using \eqref{proeq4} to substitute for \(\mathcal{E}\) in \eqref{proeq1}. It is also easy to see that if the inequality \eqref{proeq1} holds, then \(\mathcal{E}\leq-(2/3)\rho\implies\lambda\geq0\).\qed
\end{proof}

We note that in the non-vacuum case a necessary condition for stability of the MOTS is that the Weyl scalar is strictly negative on the MOTS. In the vacuum case, of course this is both necessary and sufficient for strict stability.

Indeed, if the spacetime in which the MOTS is embedded is conformally flat, i.e. \(\mathcal{E}=0\) (we could simply impose that this condition is horizon compatible and does not need to hold on he entire spacetime), then the MOTS are always unstable for positive \(\rho\). A well known case of this consideration is that of MOTS contained in TLM in the Oppenheimer-Snyder dust collapse.

Another spacetime of wide ranging interests, which falls in the class of spacetimes considered in this work, is the Lemaitre-Tolman-Bondi (LTB) spacetimes. It was shown in \cite{shef2} that DH in these spacetimes must satisfy \(\rho<-6\mathcal{E}\), and was further demonstrated that stability of MOTS implies that the horizon is a DH. It is clear that DH in these spacetimes will necessarily satisfy the following upper bound on the area:

\begin{eqnarray*}
\mathbb{A}<\frac{1}{\rho}.
\end{eqnarray*}
We see that \eqref{proeq1} is a refinement of the above inequality. This therefore suggests that there may be DH in these spacetimes containing unstable MOTS. In particular, MOTS that lie in a DH with area within the range

\begin{eqnarray*}
\frac{1}{2\rho}<\mathbb{A}<\frac{1}{\rho},
\end{eqnarray*}
will be unstable.

Proposition \ref{prop1} suggests that the non blow-up of the horizon is necessary for stability and it is ensured by \(\mathcal{E}<0\). Here we see that the cut-off on the area is rather a criteria and/or consequence of stability. Here we have specified the upper bound on the horizon area as the inverse of the local energy density.

Now, we had earlier pondered what implications, in the case of a MTS within context of our formulation, would the sign of the variation of the sheet expansion have on stability of the MTS. As an eigenvalue problem, in a general 1+1+2 spacetime, the variation of \(\phi\) should give rise to principal eigenvalue of the form

\begin{eqnarray}\label{dfr}
\bar{\lambda}=\sqrt{2}\lambda-\mbox{extra term},
\end{eqnarray}
where the ``extra term" is a linear combination of any two of the three scalars \(\Theta,\Sigma,\phi\), and some scalars formed from sheet terms. Hence, it appears that the sign of \(\phi(\Psi)\) does carry information about stability of the MOTS. Of course, for the present work, stability check is straightforward because of the form \(\lambda\) takes, and therefore, there is no need for \(\bar{\lambda}\). However, this could be useful in more general settings.

Let us now prove the following:

\begin{theorem}
Let \(\mathcal{S}\) be a MOTS contained in a horizon embedded in a perfect fluid. Then, the principal eigenvalue is non-zero.
\end{theorem}
\begin{proof}
Suppose the principal eigenvalue is zero. Using the propagation equations \cite{cc1}

\begin{eqnarray*}
\begin{split}
\frac{2}{3}\hat{\Theta}-\hat{\Sigma}&=\frac{3}{2}\phi\Sigma,\\
\hat{\phi}&=-\frac{1}{2}\phi^2+\left(\frac{1}{3}\Theta+\Sigma\right)\left(\frac{2}{3}\Theta-\Sigma\right)-\frac{2}{3}\rho-\mathcal{E},
\end{split}
\end{eqnarray*}
it is easily checked that \(\lambda=\hat{\chi}\). Therefore if \(\lambda=0\), we have that \(\hat{\chi}=0\). Then, we have that \(\mathcal{C}=1\), and the horizon is a DH. From the definition of \(\mathcal{V}^{\mu}\), we have that \(u^{\mu}\) is everywhere orthogonal to the horizon, i.e. a constant time slice. Now, of course all evolution equations on such a hypersurface are constraints. The constraint from the evolution of the local energy density in this case is (again, see \cite{cc1} for the full set of equations)

\begin{eqnarray*}
0=\Theta\left(\rho+p\right),
\end{eqnarray*}
so that \(\Theta\) must vanish of the DH since \(\rho+p\neq0\). We therefore have, from the evolution equations for \((2/3)\Theta-\Sigma\) and \(\phi\), the respective constraints

\begin{eqnarray*}
\begin{split}
0&=-\left(\rho+p\right)+\phi\left(\mathcal{A}-\frac{1}{2}\phi\right)+\lambda,\\
0&=\phi\left(\mathcal{A}-\frac{1}{2}\phi\right),
\end{split}
\end{eqnarray*}
from which we have (by setting \(\lambda=0\)) \(\rho+p=0\), again clearly not possible on a DH. Hence, \(\lambda\) cannot be zero.\qed
\end{proof}

Actually, it can be easily checked that the above result extends to MOTS in any LRS II solution.

\begin{corollary}
A MOTS \(\mathcal{S}\) contained in a horizon embedded in a perfect fluid, is either unstable or strictly stably outermost. In particular, \(\mathcal{S}\) is stably outermost if an only if the horizon area satisfies

\begin{eqnarray}\label{nbay}
\mathbb{A}<\frac{1}{2\rho}.
\end{eqnarray}
\end{corollary}


\subsection{Comments on the topology of the MOTS}


Here we provide some comments related to the topology of stable MOTS in perfect fluids and the accompanying implications. Additional analysis is carried out in the case that one extends to imperfect fluids.  

In LRS \textcolor{red}{II} solutions, the 2-surfaces under consideration are allowed to have spherical, flat (including toroidal) as well as hyperbolic topologies. Let us integrate \eqref{proeq2} over the area of \(\mathcal{S}\), and use the Gauss-Bonnet theorem (noting that \(\mathcal{S}\) is closed) so that we have

\begin{eqnarray}\label{topo1}
2\pi\underline{\chi}(\mathcal{S})=\lambda[\mathcal{S}]+\int_{\mathcal{S}}\rho,
\end{eqnarray}
with \(\underline{\chi}(\mathcal{S})\) denoting the Euler characteristics of \(\mathcal{S}\). Hence, for a stably outermost MOTS, \(\lambda[\mathcal{S}]>0\) and so \(\mathcal{S}\) can only have a spherical topology. (Note that the second term on the RHS of \eqref{topo1} ie non-negative by assumption (\(\rho\geq0\implies\)) DEC holds.) This topological result is a particular case of the following more general result due to Galloway \cite{gal2}, extending to higher dimensions:

\begin{theorem}[Galloway]
Let \((V^n,h,\mathcal{K})\) be an initial data set satisfying the dominant energy condition (DEC), \(\rho\geq J_{\mu}J^{\mu}\). If \(\tilde{\Sigma}^{n-1}\) is an outermost MOTS in \((V^n,h,\mathcal{K})\), then \(\tilde{\Sigma}^{n-1}\) admits a metric of positive scalar curvature.
\end{theorem}
(Of course in case \(n=3\), \(\tilde{\Sigma}\) is a topological 2-sphere.) However, we would like to know if the causal character of the horizon containing the MOTS is restricted by the stability criteria and/or the topology. Of course we can immediately rule out stable pressureless TLM. We can however completely rule out the possible existence of a stable MOTS in a TLM embedded in a perfect fluid: for a TLM in LRS II, the area satisfies

\begin{eqnarray*}
\mathbb{A}>\frac{1}{\rho-p}.
\end{eqnarray*}
Comparing the above to \eqref{nbay} we see that the MOTS are stable provided that \(\rho-p>2\rho\implies-\left(\rho+p\right)>0\), which is not possible given that here the NEC is assumed (also implies the weak energy condition strictly \(\rho+p>0\)). We therefore state this as the following non-existence result

\begin{proposition}\label{prot1}
A TLM embedded in a perfect fluid cannot contain a stable MOTS.
\end{proposition}

Also, notice how \eqref{topo1} allows us to rule out the other topologies if we assume stability: \(\underline{\chi}(\mathcal{S})=\) (\textit{resp.} \(<0\)) would require that \(\lambda=\) (\textit{resp.} \(<-\rho\)).

Now, suppose that we wish to extend our considerations to imperfect fluids with \(\Pi=0\), so that \(J\not\equiv0\). If \(Q<0\), there is no problem here since \(\rho-Q\) is always positive. (Of course this means that we can immediately rule out stable non-spherical 2-surface geometries.) We will show the following.

\begin{theorem}\label{goto1}
Any stable MOTS \(\mathcal{S}\) in an imperfect LRS II fluid with vanishing stress anisotropy is a topological sphere.
\end{theorem}
\begin{proof}
If \(Q\neq0\), then, \eqref{topo1} becomes

\begin{eqnarray}\label{topo2}
2\pi\underline{\chi}(\mathcal{S})=\lambda[\mathcal{S}]+\int_{\mathcal{S}}(\rho-Q).
\end{eqnarray}
It is then quite easy to show that \(Q\leq\rho\), so that if \(\mathcal{S}\) is stable, then, \(\underline{\chi}(\mathcal{S})>0\): it can be checked that the denominator of \eqref{proeq1} acquires the additional term \(-2Q\). It is not difficult to see that for S to be stable, it is necessary that \(Q\geq\mathcal{E}\). From \eqref{proeq4} we have the desired area bound:

\begin{eqnarray*}
\mathbb{A}=\frac{3}{2\left(\rho-3Q\right)},
\end{eqnarray*}
from which it follows that \(Q\leq\rho\).\qed
\end{proof}

Notice that in the above result, we do not impose the DEC. However, the DEC  in this case reduces to

\begin{eqnarray}\label{deco}
\rho\geq Q^2-\rho^2,
\end{eqnarray}
which is always satisfied for \(Q\leq\rho\). Furthermore, there is no sign fixing on the horizon and hence the above result would appear to  hold true for arbitrary signature. However, just as before, we can completely rule out the existence of stable MOTS in a TLM embedded in an imperfect LRS fluid, where the argument follows similarly as the perfect fluid case:

\begin{proposition}\label{prot2}
A TLM embedded in an imperfect fluid with vanishing anisotropy cannot contain a stable MOTS.
\end{proposition}

We point out that equality holds for \(Q\leq\rho\) if an only if the horizon in which the MOTS lie is isolated, with the particular equation of state \(\rho=p\). For pressure-less imperfect fluids with no anisotropy, we see that for stable MOTS, away from isolation (more specifically, in a DH), \(Q\) is bounded  from both sides:

\begin{eqnarray*}
\mathcal{E}<Q<\rho.
\end{eqnarray*}
(For a DH we assume the infalling of radiation so that \(Q>0\). Using \eqref{x2})

\begin{eqnarray}\label{q1}
\mathcal{E}<Q<-6\mathcal{E}.
\end{eqnarray}
Observe that \(\mathcal{E}\) cannot be zero, and it is clear that we must have \(\mathcal{E}<0\). It therefore follows that on any MOTS contained in a DH in an LRS II spacetime with no anisotropy, and on which the pressure \(p\) vanishes and \(Q>0\), \(\mathcal{E}<0\). This is of course expected for otherwise, the horizon cannot be a DH, i.e. the denominator of \(\mathcal{C}\) is non-negative.

Indeed, if the sign of \(Q\) is known, then, \eqref{q1} provides us with a way to quickly check the (non)-existence of stable MOTS in a pressure-less scenario. For example, if for some constant \(\bar{c}\leq-6\) we have that \(Q=\bar{c}\mathcal{E}\) on a DH \(\mathcal{H}\), then, \(\mathcal{H}\) cannot contain a stable MOTS. For the case of vanishing acceleration, we can use the field equations and the commutator of the dot and hat derivatives to show that this condition is precisely

\begin{eqnarray*}
\hat{\mathcal{E}}\geq-\phi\left(\frac{9}{2}\mathcal{E}+\frac{2}{3}\Theta^2\right),
\end{eqnarray*}
where the number \(c\) is just (note that \(\phi\neq0\) on a DH, and \(\mathcal{E}\neq0\) as discussed above)

\begin{eqnarray*}
\bar{c}=\frac{1}{\phi\mathcal{E}}\left(2\hat{\mathcal{E}}+\phi\left(3\mathcal{E}+\frac{4}{3}\Theta^2\right)\right).
\end{eqnarray*}

We also point out that the ruling out of stable surfaces of other topologies in the perfect fluid case also holds here with \(Q\neq0\). In fact, we can make the following stronger claim: \textit{Horizons in LRS II solutions with no anisotropy, of the type considered in this work, cannot be foliated by non-spherical MOTS (even unstable ones).} To see this, we observe that for both the perfect and imperfect case with no anisotropy (without assuming any sign restriction on \(\lambda\)), the area \(\mathbb{A}\) blows up if \(\underline{\chi}(\mathcal{S})=0\) (an infinitely large black hole which we can rule out) and \(\mathbb{A}<0\) (which is clearly impossible) if \(\underline{\chi}(\mathcal{S})<0\). However, if we throw in anisotropy, then such foliation is in principle possible for \(\Pi>0\). This is due to the introduction of the \((1/2)\Pi\) term in \(\lambda\), a contribution that comes from the Gaussian curvature of the MOTS.

Let us end this section with the following remark:

\begin{remark}[Remark 2.]
\textit{While we have neither explicitly constructed MOTS nor seen explicit examples of them in LRS solutions with at least one of \(\xi\) or \(\Omega\) non-vanishing, to our understanding we are not able to rule out their existence. Therefore, one could attempt to generalize these results to LRS solutions with rotation and spatial twist if they do admit MOTS. However, there are two points to note:}

\begin{enumerate}
\item \textit{Firstly, rotation and spatial twist are obstructions to nonminimal MOTS, so that, for the MOTS we consider, if either one of the rotation or spatial twist is zero, the MOTS can only foliate minimal horizons as was earlier alluded to. The problem is, since the main focus of this is relating horizon area bound and stability, the minimal cases are not very much of interest here. Furthermore, for the eigenvalue, the only extra contribution from non LRS II quantities can only possibly be from the scalar curvature of the MOTS. In particular, this will be a factor of the square of the rotation scalar \(\Omega\), and so the stability of the minimal MOTS is determined by the difference between \(\mathcal{E}\) and the factor of the square of the rotation;}

\item \textit{Secondly, with respect to the topological considerations for this subsection, all the results applicable to the null case will apply if we include rotation and spatial twist, as the contributing term is encoded in \(\underline{\chi}(\mathcal{S})\).}
\end{enumerate} 
\end{remark}


\section{The growth of some covariant variables on horizons containing unstable MOTS}\label{bound}


Let us now look at some supplementary results where the magnitude of the function \(\mathcal{C}\) provides sign constraints on the propagation of the electric Weyl scalar and energy density. As a consequence, the following result will be demonstrated.

\begin{proposition}\label{thgr}
Let \(\mathcal{H}\) be a DH in a perfect fluid with \(\mathcal{C}>1\) and \(\Theta\geq0\), and let \(\mathcal{S}\) be a MOTS in \(\mathcal{H}\). Suppose the isotropic pressure is sufficiently small so that first order derivatives are negligible. Then, for small \(t>0\), the variation of the outgoing null expansion of \(\mathcal{S}=\mathcal{S}_0\) is less than that of \(\mathcal{S}_t\).
\end{proposition}

To prove the above result, we will go through a series of intermediate results - and accompanying discussions - which transparently demonstrates the interplay between the curvature quantities appearing in \(\lambda\) and their propagation along the direction along which the MOTS is varied, i.e., \(e^{\mu}\). (Notice that, if the pressure is sufficiently small relative to the local energy density, then, on a DH we must have \(\mathcal{E}<0\). Therefore, whenever the smallness of the pressure is assumed, it will be understood that \(\mathcal{E}\) is strictly negative.)

Before proceeding, we make a few observations. We observe that the following set of statements is true for a MOTS contained in a DH in any LRS II solution:

\begin{eqnarray*}
\begin{split}
0<\mathcal{C}<1&\implies\lambda>0\quad\mbox{(strictly stable)};\\
\mathcal{C}>1&\implies\lambda<0\quad\mbox{(unstable)}.
\end{split}
\end{eqnarray*}
Indeed, we know that \(\lambda\neq0\) on a MOTS contained in \(\mathcal{H}\) (\(\lambda=0\implies \mathcal{C}=1\)), and hence, while \(\lambda\) increases along \(e^{\mu}\), the sign stays fixed. This also provides a picture of the evolution of an unstable MOTS: if a MOTS is unstable, then, one can tell its future evolution - assuming all reasonable energy conditions are satisfied - by the sign of \(\mathcal{L}_l\chi\), i.e., whether it evolves to foliate a TLM or a DH.

Now, Proposition \ref{prot1} confirms a widely accepted view that only stable MOTS with globally spacelike world tubes (i.e., they foliate a DH) can be identified as black hole boundaries. The conditions on the magnitude, as it relates to stability of the MOTS we have considered in this work, have the following implication: On a MOTS, if the quantity \(\mathcal{C}\) takes a value in the interval \(0<\mathcal{C}<1\), then, that MOTS will evolve into a MOTT that is a suitable black hole boundary.

Let us now begin by taking the Lie derivative of \eqref{proeq4} along the tangent vector field \(\mathcal{V}^{\mu}\) to obtain

\begin{eqnarray*}
\frac{1}{3}\mathcal{L}_{\mathcal{V}}\rho-\mathcal{L}_{\mathcal{V}}\mathcal{E}=-\frac{1}{2\mathbb{A}^2}\mathcal{L}_{\mathcal{V}}\mathbb{A}.
\end{eqnarray*}
As is known, if the horizon is a TLM, the area is decreasing and hence \(\mathcal{L}_{\mathcal{V}}\mathbb{A}<0\). And if it is a DH, the area is non-decreasing and \(\mathcal{L}_{\mathcal{V}}\mathbb{A}\geq0\). We therefore have that

\begin{eqnarray*}
\frac{1}{3}\mathcal{L}_{\mathcal{V}}\rho-\mathcal{L}_{\mathcal{V}}\mathcal{E}>(resp. \leq)\ 0,
\end{eqnarray*}
for a TLM (\textit{resp.} DH). Explicitly, the above expression becomes

\begin{eqnarray*}
\left(1-\mathcal{C}\right)\left(\frac{1}{3}\dot{\rho}-\dot{\mathcal{E}}\right)+\left(1+\mathcal{C}\right)\left(\frac{1}{3}\hat{\rho}-\hat{\mathcal{E}}\right)>(resp. \leq)\ 0,
\end{eqnarray*}
which, on the horizon, simplifies to

\begin{eqnarray*}
\phi\left(\frac{1}{2}\left(1-\mathcal{C}\right)\left(\rho+p\right)+3\mathcal{C}\mathcal{E}\right)>(resp. \leq)\ 0,
\end{eqnarray*}
where we have used the evolution and propagation equations \cite{cc1}

\begin{subequations}
\begin{align}
\dot{\mathcal{E}}-\frac{1}{3}\dot{\rho}&=-\left(\frac{2}{3}\Theta-\Sigma\right)\left(\frac{1}{2}\left(\rho+p\right)-\frac{3}{2}\mathcal{E}\right),\label{fent1}\\
\hat{\mathcal{E}}-\frac{1}{3}\hat{\rho}&=-\frac{3}{2}\phi\mathcal{E}.\label{fent2}
\end{align}
\end{subequations}
Away from minimality, we know that \(\phi\) is strictly positive, and hence we have the following

\begin{eqnarray}\label{ed}
\left(1-\mathcal{C}\right)\left(\rho+p\right)>(resp. \leq)\ -6\mathcal{C}\mathcal{E}
\end{eqnarray}

Now, take the case of a DH \(\mathcal{H}\) and suppose that the Weyl scalar is strictly negative. Then, clearly for \(\mathcal{E}<0\), \eqref{ed} is always true for \(\mathcal{C}>1\) (note that \(\rho+p>0\) is always true since \(\mathcal{H}\) is non-null and, \(\rho\geq0\implies\mbox{DEC}\implies\mbox{NEC}\implies\rho+p>0\)). However, for \(0<\mathcal{C}<1\), \eqref{ed} is a constraint.

As an aside, a point of consideration, where the subtlety of the above discussion is crucial, is the case of a slowly evolving horizon \cite{ib6,ib7}, usually defined in terms of the smallness of a so-called \textit{slowly evolving parameter} \(\epsilon\). It has been suggested that \cite{shef2}, at least for the class of locally rotaionally symmetric solutions, slowly evolving horizons can simply be interpreted as those for which \(|\mathcal{C}|\) is sufficiently small (in fact, this is easily seen to lead to an equivalent interpretation of slowly evolving horizons: those for which the magnitude of the principal eigenvalue of the stability operator is sufficiently small). In this case, one sees from \eqref{ed} that the WEC must be violated in this limit, which is clearly not possible for the cases considered here. This would therefore suggest the following: \textit{A DH \(\mathcal{H}\), embedded in a perfect fluid and containing a stable MOTS (the ``stable MOTS" requirement can be relaxed to the requirement that the Weyl scalar is strictly negative) cannot slowly evolve for a ``sufficiently large" \(\rho\) on \(\mathcal{H}\).} 

The necessity of the qualification that the energy density be large enough is easily seen from \eqref{ed} as the inequality would hold if \(\rho\) is small enough. Take as an example the LTB solution, which is pressureless with its metric in local coordinates given as 

\begin{eqnarray}\label{lt1}
ds^2=-d\bar{\tau}^2+\frac{\mathcal{R}'^2}{1-\frac{2\mathcal{M}}{\bar{r}}}d\bar{r}^2+\mathcal{R}^2dS^2,
\end{eqnarray}
with the prime denoting differentiation with respect to the radial coordinate \(\bar{r}\), where \(\mathcal{M}=\mathcal{M}(\bar{r})\) corresponds to the Misner-Sharp mass, \(\mathcal{R}=\mathcal{R}(\bar{\tau},\bar{r})\), and \(dS^2\) is the spatial 2-surfaces metric. The coordinate expression of the energy density is given as

\begin{eqnarray}\label{lt2}
\rho=\frac{\left(\bar{r}^3\mathcal{M}\right)'}{\mathcal{R}^2\mathcal{R}'}.
\end{eqnarray}
It is clear that, for a fixed \(\bar{r}\) if the change in the mass \(\mathcal{M}\) is sufficiently small, so is \(\rho\), which implies that \(\mathcal{C}\) is also very small. Of course, for sufficiently small \(\mathcal{C}\), \(\rho\) is sufficiently small, which implies \(\mathcal{M}'\) is very small. 

MOTS in these horizons are stable and the horizon always slowly evolves. Detailed considerations of slowly evolving horizons in the LTB solution can be found in \cite{ib7}, where for the cases considered on a DH, the slowly evolving parameter takes the form

\begin{eqnarray*}
\epsilon^2=8\frac{\mathcal{M}'}{\mathcal{R}'-\mathcal{M}'},
\end{eqnarray*}
from which it is clear that \(\mathcal{M}'\) being sufficiently small implies \(\epsilon\) is sufficiently small.

Note that if the stability requirement for the MOTS is relaxed, then in general the sign of \(\mathcal{E}\) is not constrained. However, if \(0<\mathcal{C}<1\), then \(\mathcal{E}\) has to be strictly negative (not considering the slowly evolving case).

The next set of results examines the ``growth" of the Weyl curvature and local energy variables. Let us begin with the following result.

\begin{proposition}\label{prop20}
Suppose \(\mathcal{H}\) is DH embedded in a perfect fluid, with \(\mathcal{C}>1,\mathcal{A}\geq0\) and \(\Theta\geq0\). Then, at least one of the following is true:

\begin{enumerate}
\item The isotropic pressure decreases along \(u^{\mu}\); or

\item The local energy density decreases along \(e^{\mu}\).
\end{enumerate}
\end{proposition}
\begin{proof}
Taking the Lie derivative of \eqref{x1} along \(\mathcal{V}^{\mu}\), and noting that \(\mathcal{C}\) is constant, we obtain

\begin{eqnarray*}
\left(\mathcal{L}_{\mathcal{V}}\mathcal{L}_l\chi\right)\left(\mathcal{L}_k\chi\right)=\left(\mathcal{L}_{\mathcal{V}}\mathcal{L}_k\chi\right)\left(\mathcal{L}_l\chi\right).
\end{eqnarray*} 
Since the DEC is always implicit as it is assumed that \(\rho\geq0\), the NEC always holds so that \(\mathcal{L}_{\mathcal{V}}\mathcal{L}_k\chi<0\). And since \(\mathcal{L}_k\chi<0\) and \(\mathcal{L}_l\chi<0\), we have that

\begin{eqnarray}\label{fir1}
\mathcal{L}_{\mathcal{V}}\mathcal{L}_l\chi\leq0.
\end{eqnarray} 
Using \(\mathcal{L}_l\chi=\rho-p-1/\mathbb{A}\), the above inequality can be explicitly written as (where we again use the fact that \(\mathcal{L}_{\mathcal{V}}\mathbb{A}\geq0\))

\begin{eqnarray}\label{fir2}
\left[\left(\mathcal{C}-1\right)\Theta+\left(1+\mathcal{C}\right)\mathcal{A}\right]\left(\rho+p\right)+\left(1+\mathcal{C}\right)\hat{\rho}+\left(\mathcal{C}-1\right)\dot{p}\leq0.
\end{eqnarray} 
Indeed, since \(\mathcal{A}\geq0\) and \(\Theta\geq0\), for \(\mathcal{C}>1\) it follows that \(\hat{\rho}\) and \(\dot{p}\) cannot simultaneously be positive.\qed
\end{proof}

The following corollary is a consequence of Proposition \ref{prop20}.

\begin{corollary}\label{cor1}
Let \(\mathcal{H}\) be a DH with \(\mathcal{C}>1\), \(\mathcal{A}\geq0\) and \(\Theta\geq0\), embedded in a perfect fluid. If \(\mathcal{L}_{l}\chi\) is monotone along \(e^{\mu}\) and the Weyl scalar is strictly negative on \(\mathcal{H}\), then, the local energy density must decrease along \(e^{\mu}\).
\end{corollary}
\begin{proof}
We prove this corollary by showing that, under the assumptions, \(\dot{p}\) is strictly positive. We expand \eqref{fir1} as

\begin{eqnarray*}
\left(1-\mathcal{C}\right)\left(\dot{\chi}-\hat{\chi}\right)^{\cdot}+\left(1+\mathcal{C}\right)\widehat{\left(\dot{\chi}-\hat{\chi}\right)}\leq0,
\end{eqnarray*}
from which, for \(\mathcal{C}>1\), we have

\begin{eqnarray}\label{fir10}
0<\frac{\mathcal{C}+1}{\mathcal{C}-1}\leq\frac{\left(\dot{\chi}-\hat{\chi}\right)^{\cdot}}{\widehat{\left(\dot{\chi}-\hat{\chi}\right)}}
\end{eqnarray}
Clearly, the numerator and denominator of \eqref{fir10} must have the same sign. And since \(\mathcal{L}_{l}\chi\) is monotone along \(e^{\mu}\), the denominator is negative \(\implies\) numerator must also be negative, and therefore we have that

\begin{eqnarray*}
\dot{p}>-3\phi\mathcal{E}+\left(\Theta+\phi\right)\left(\rho+p\right),
\end{eqnarray*}
where we have used the evolution equation \eqref{fent1}. Since \(\mathcal{E}\) is strictly negative, the right hand side of the above expression is positive, from which it follows that \(\dot{p}>0\). Hence, by Proposition \ref{prop20}, \(\hat{\rho}\) must be negative.\qed
\end{proof}

\begin{proposition}\label{prop22}
Let \(\mathcal{H}\) be a DH embedded in a perfect fluid with \(\mathcal{C}>1\) and \(\Theta\geq0\), and suppose the Weyl scalar is strictly negative. If the pressure is sufficiently small so that first order derivatives are negligible, then the Weyl distortion increases along \(e^{\mu}\).
\end{proposition}
Proposition \ref{prop22} follows as a corollary to the following proposition.

\begin{proposition}\label{prop23}
Let \(\mathcal{H}\) be a DH embedded in a perfect fluid with \(\mathcal{C}>1\) and \(\Theta\geq0\), and suppose the Weyl scalar is strictly negative. If the isotropic pressure is sufficiently small so that first order derivatives are negligible, then \(\mathcal{L}_{l}\chi\) is increasing along \(e^{\mu}\).
\end{proposition}
\begin{proof}
From Proposition \ref{prop20}, if the isotropic pressure is small enough so that its derivative along \(u^{\mu}\) and \(e^{\mu}\) are negligible, then \(\hat{\rho}\) is negative for \(\mathcal{C}>1\). It is also seen from the propagation for the isotropic pressure \cite{cc1}

\begin{eqnarray*}
\hat{p}=-\mathcal{A}\left(\rho+p\right),
\end{eqnarray*}
that \(\mathcal{A}\) is neglible as well so that \(\mathcal{H}\) is taken as non-accelerating. Assume to the contrary that \(\mathcal{L}_{l}\chi\) is decreasing along \(e^{\mu}\). By Corollary \ref{cor1}, \(\mathcal{L}_l\chi\) is also decreasing along \(u^{\mu}\) and so we have that

\begin{eqnarray*}
-3\phi\mathcal{E}+\left(\Theta+\phi\right)\left(\rho+p\right)<0.
\end{eqnarray*}
For the above inequality to hold, we must have that \(\mathcal{E}\) is positive, which of course contradicts the assumption that \(\mathcal{E}<0\), and therefore \(\mathcal{L}_l\chi\) must instead be increasing along \(u^{\mu}\). And since the signs of the derivatives of \(\mathcal{L}_l\chi\) along \(u^{\mu}\) and \(e^{\mu}\) must be the same from \eqref{fir10}, \(\mathcal{L}_l\chi\) must be increasing along \(e^{\mu}\).\qed
\end{proof}

\noindent\textit{Proof of Proposition \ref{prop22}}: To prove Proposition \ref{prop22}, we note that the statement that \(\mathcal{L}_{l}\chi\) is increasing along \(e^{\mu}\) can be written explicitly as

\begin{eqnarray}\label{df}
\hat{\rho}>3\phi\mathcal{E}.
\end{eqnarray}
Now, we can rewrite \eqref{fent2} as

\begin{eqnarray*}
\hat{\mathcal{E}}=-\frac{1}{6}\hat{\rho}+\frac{1}{2}\left(\hat{\rho}-3\phi\mathcal{E}\right),
\end{eqnarray*}
so that the above expression is strictly positive, where we use \eqref{df} and the fact that \(\hat{\rho}<0\).\qed

Indeed, it then follows that

\begin{eqnarray*}
-\frac{2}{3}\hat{\rho}-\hat{\mathcal{E}}>0,
\end{eqnarray*}
which verifies the claim of Proposition \ref{thgr}.


\section{Conclusion}\label{8}


In this work we have carried out a careful analysis of stability of cross sections of certain black hole horizons and their relationship to the horizon area and the topology. We employed the 1+1+2 covariant formulation and limited our focus to the case of LRS II spacetimes with a perfect fluid geometry, and where the topology is concerned, we extended to the imperfect case with non-vanishing heat flux. 

Firstly, it was shown that the MOTS stability operator is self-adjoint, since the linear term vanishes as it is acted on by the vanishing rotation one-form on the MOTS. (The vanishing of this one-form is true for more general LRS spacetimes which include rotation and/or spatial twist. Hence, if these solutions admit MOTS, the operator is self adjoint.) In spherical symmetry, the stability operator simplifies significantly and the constant function \(Psi = const\) is the eigenfunction of the principal eigenvalue. An upper area bound on the horizon area containing stable MOTS was obtained, with a strict inequality demonstrated, after it was shown that the eigenvalue of the operator is non-zero. Stable MOTS in timelike membranes were also ruled out, even with the introduction of heat flux. These findings confirms the rarity of stable MOTS as well as DH that one may associate with black holes.

We briefly discussed the topological restrictions on the MOTS and in the case where one extends to an imperfect fluid with non-vanishing heat flux, stability imposes a relationship between the flux and the local energy density, which ensure spherical topology of the MOTS. This relationship, in the case of positive heat flux leads to a non-existence result, for which the precise condition in terms of the spacetime variables is provided. In particular, it is shown that for a non-accelerating imperfect LRS II solution with vanishing anisotropy, if on a DH one has that

\begin{eqnarray*}
\hat{\mathcal{E}}\geq-\phi\left(\frac{9}{2}\mathcal{E}+\frac{2}{3}\Theta^2\right),
\end{eqnarray*}
then, the DH cannot contain a stable MOTS.

Finally, since the eigenvalue is a linear combination of the local energy density and the electric Weyl scalar and can be written as the directional derivative of the outgoing null expansion scalar along the unit normal to the MOTS, we considered how these quantities behave along this direction. There are values of the function \(\mathcal{C}\) (even for the case of DH) - used to determine the horizon signature - which imply that the MOTS are necessarily unstable. Under certain mild assumptions, it is shown that for an unstable MOTS contained in a DH the local energy decreases along the unit normal to the MOTS while the Weyl distortion increases. It then follows that, for small \(t>0\), the variation of the outgoing null expansion of \(\mathcal{S}=\mathcal{S}_0\) is less than that of \(\mathcal{S}_t\).

The stability analysis here could be extended to more general MOTS embedded in general 1+1+2 spacetimes. We briefly mentioned this in Remark 1 and propose a variation of the stability that could be performed. However, as a start, one could study the stability using the formulation as developed by Anderson \textit{et al.}. In this case one would certainly encounter much more difficulty, not only due to the presence of more terms for the operator, but also the fact that \(s_{\mu}\) will take a much more complicated form. In addition, the Laplacian will not vanish and will take a very messy form. As a first step towards such a generalisation one can impose on the spacetime the conditions for which \(s_{\mu}=0\) to bring the operator into a self-adjoint form. In fact, one can even replicate the form of the operator in this work by further assuming that scalars can be expanded up to second order on the MOTS so that second surface derivatives are negligible. 

There is another reason for interest in this paper, which is linked to an ongoing work by the same authors. The aim is to try to look at MOTS in perturbed LRS spacetimes and the behaviour of the horizons they foliate. Indeed, the one-form \(s_{\mu}\), as well as the Laplacian, is in general non-zero. It is of interest to see whether there are constraints on the contributing terms to the Euler characteristic \(\underline{\chi}(\mathcal{S})\) of MOTS in the perturbed spacetime and how these constraints influence the form of the operator. Additionally, it will be a good exercise to study and obtain conditions (and implications) under which the form of the stability operator remains invariant under the perturbation of the underlying spacetime. In fact, this problem is intricately linked to the problem of finding interior MOTS that have recently been investigated (see \cite{ib2,ib3,ib8,ib9}). In \cite{ib8,ib9}, the authors studied MOTS in the interior of a Schwarzschild black hole by adapting to the Painlev\`{e}-Gullstrand ``\textit{horizon penetrating coordinates}". In these coordinates, the constant `time" slices are flat but non-static. (The flatness property allows one to study surfaces using analogous approaches in the standard Euclidean case and the non-staticity means that one also captures non-minimal MOTS). The transformed metric takes a form with a twist term. This therefore would suggest that, in order to study MOTS not lying in the \(r=2m\) hypersurface, from context of the 1+1+2 approach, one may consider the transformed metric metric as a perturbation of the Schwarzschild spacetime in standard coordinates. In this case, it is then possible to conduct a detailed analytic study of these interior MOTS, as perturbation of Schwarzschild spacetime in our approach has been well developed \cite{cc2}. Furthermore, successful implementation of this task would allow us to generalise to a general LRS spacetime, which could potentially capture interior MOTS in a range of spacetimes for which details are only known of those standard MOTS on the boundary, i.e. LTB, FLRW, etc.


\section*{Acknowledgement}

The authors would like to thank the anonymous referees for valuable comments and suggestions, which have significantly improved the quality and clarity of the manuscript. AS acknowledges that this work was supported by the Institute for Basic Science (IBS-R003-D1). PKSD is supported by the First Rand Bank, South Africa.

\end{document}